\documentclass[reprint,aps,superscriptaddress,twocolumn,longbibliography]{revtex4-1}

\usepackage{times}

\usepackage{xr}
\usepackage{filecontents}
\usepackage{amsthm}
\usepackage{amsmath,tabu}
\usepackage{amssymb}
\usepackage[pdftex]{graphicx}
\usepackage{accents}

\usepackage{bbold}

\usepackage{color}

\usepackage{hyperref}

\hypersetup{
  colorlinks = true
}

\definecolor{darkred}  {rgb}{0.5,0,0}
\definecolor{darkblue} {rgb}{0,0,0.5}
\definecolor{darkgreen}{rgb}{0,0.5,0}

\hypersetup{
  urlcolor  = darkblue,     
  linkcolor = darkblue,     
  citecolor = darkgreen,    
  filecolor = darkred       
}

\usepackage{chemarrow}

\makeatletter
\newcommand*{\frightarrowfill@}{
  \arrowfill@\relbar\relbar\chemarrow}
\renewcommand{\xrightarrow}[2][]{\ext@arrow 0369\frightarrowfill@{\vspace{-1mm}#1}{#2}}
\makeatother

\usepackage[caption=false]{subfig}

\newtheorem{theorem}{Theorem}

\newtheorem{main}{Main Result}
\newtheorem*{main1*}{Main Result 1}
\newtheorem*{main2*}{Main Result 2}
\newtheorem{lemma}{Lemma}

\newtheorem{corollary}{Corollary}
\newtheorem{definition}{Definition}
\newtheorem{example}{Example}

\newcommand\word[1]{\mathbf{#1}}
\newcommand{\mc}[1]{\mathcal{#1}}

\newcommand{\A}{\mc{A}}
\newcommand{\B}{\mc{B}}
\newcommand{\C}{\mc{C}}
\newcommand{\D}{\mc{D}}
\newcommand{\F}{\mc{F}}
\newcommand{\Q}{\mc{Q}}
\newcommand{\T}{\mc{T}}

\newcommand{\W}{\mc{W}}
\newcommand{\X}{\mc{X}}
\newcommand{\Y}{\mc{Y}}

\newcommand{\id}{\textrm{id}}
\newcommand{\rt}{\textrm{rt}}
\newcommand{\val}{\textrm{val}}
\newcommand{\valA}{\val_{\mathcal A}}

\newcommand{\stable}{freezable}

\begin{document}

\title{Memory effects can make the transmission capability of a communication channel uncomputable}


\author{David Elkouss}
\email{Correspondence and requests for materials should be addressed to D.E. (email: d.elkousscoronas@tudelft.nl)}
\affiliation{QuTech, Delft University of Technology, Lorentzweg 1, 2628 CJ Delft, Netherlands}
\author{David P\'erez-Garc\'ia}
\affiliation{Departamento de An\'alisis Matem\'atico and Instituto de Matem\'atica Interdisciplinar, Universidad Complutense de Madrid, 28040 Madrid, Spain}
\affiliation{ICMAT, c/ Nicol\'as Cabrera, Campus de Cantoblanco, 28049 Madrid, Spain}


\maketitle

{\bf 
Most communication channels are subjected to noise. One of the goals of Information Theory is to add redundancy in the transmission of information so that the information is transmitted reliably and the amount of information transmitted through the channel is as large as possible. The maximum rate at which reliable transmission is possible is called the capacity.
If the channel does not keep memory of its past, the capacity is given by a simple optimization problem and can be efficiently computed. 
The situation of channels with memory is less clear. Here we show that for channels with memory the capacity cannot be computed to within precision ${1}/{5}$. Our result holds even if we consider one of the simplest families of such channels -information-stable finite state machine channels-, restrict the input and output of the channel to $4$ and $1$ bit respectively and allow $6$ bits of memory.
}

\section*{Introduction}
The need to manipulate large amounts of information is one of the main characteristics of our society. It is crucial to protect the information against noise and errors in order to ensure its reliable transmission and long term storage. It is important also to do so in the optimal way so that communication channels transmit and memories store trustworthily as much information as possible. This problem motivated Shannon, already in 1948, to develop the theory of communications \cite{Shannon_48}. The natural problem that Shannon posed is, given a noisy communication channel, find the maximum rate of information it can transmit with an arbitrarily small error.

In an ingenuity tour de force, he proved that for  channels that keep no memory of their past uses (called memoryless), this quantity -the capacity of the channel- defined in such operational way, has a simple entropic expression. It coincides with the maximization, on all inputs to the channel, of the so-called mutual information between input and output in one single use of the channel. This coding result was complemented  years later by the Blahut-Arimoto (BA) algorithm \cite{Blahut_72,Arimoto_72},  which allows to efficiently approximate the capacity of any memoryless channel within any desired precision. 

The situation for channels with memory is less clear. Regarding coding theorems, more and more general classes of channels were successfully dealt with \cite{Dobrushin_63,Ahlswede_68,Winkelbauer_71,Kieffer_74} culminating in the generalized capacity formula \cite{Verdu_94}. In this last work, Verdu and Han derived a generalization of Shannon's coding theorem which essentially makes no assumption regarding the structure of the channel. 
When it comes to algorithms that approximate the capacity, despite considerable effort, the situation is nowadays less successful. Even if we restrict to the simplest case of channels with memory, the so-called finite state machine channels (FSMCs), the problem remains open. There is a rich literature dealing with particular cases (see e.g.  \cite{Gallager_68, Mushkin_89, Goldsmith_96, Pfister_01, Arnold_06, Sharma_01, Kavcic_01, Holliday_06, Pfister_10, Vontobel_08, Han_15}). However, these results do not address FSMCs in full generality or can not guarantee the precision of the result. 

It is the main aim of this work to show that an algorithm that computes approximately the capacity of an arbitrary FSMC cannot exist.  Since computable functions are exactly those that can be computed by an algorithm, this is equivalent to show that any function that approximates sufficiently the capacity of any FSMC must necessarily be uncomputable. 

\section*{Results}
\subsection*{Notation and main statements}
Aiming at an impossibility result, the simpler the family of channels we consider, the stronger the result. This is why we consider FSMCs. The same result hence holds true for any more general family of channels with memory.

In order to be precise, a FSMC with $n$ possible input symbols (the number of possible output symbols will be always $2$) and $m$ possible states in the memory is determined by \cite{Gallager_68} a set of conditional probability assignments. 
The set of conditional probability assignments $p(y,s|x,s')$  describes the probability of output symbol $y$ and transition to state $s$ in the memory if the FSMC is in state $s'$ and gets $x$ as input. Moreover, we will only consider FSMCs in which the initial state is fixed and known to the sender and receiver. We denote the initial state by $s_0$.

To avoid problems of approximating $p(y,s|x,s')$ we will only consider FSMCs for which the probability assignments $p(y,s|x,s')$ are rational numbers. Moreover, we will only consider FSMCs for which $p(y,s|x,s')$ are in product form $p(y|x,s')p(s|x,s')$ and which are information stable. Information stable channels are one of the simplest classes of channels with memory. For these channels the capacity is given by the limit of the mutual information rate \cite{Hu_64} and it is not necessary to consider the most general capacity formula \cite{Verdu_94}.

Our main result can then be stated as: 

\begin{main}\label{main1}
Any function that on input the set of probability assignments $\{p(y|x,s'), p(s|x,s')\}_{s,y,x,s'}$ of an information stable FSMC $\mathbf{N}$ with $10$ input symbols and $62$ states, outputs a rational number $c$ so that the capacity of $\mathbf{N}$ verifies
\begin{equation}
\left|C(\mathbf{N})-c\right|\le \frac{1}{5}\; ,
\end{equation}
must be uncomputable.
\end{main}

It is obvious that the same result then holds for $n$ input symbols and $m$ states as long as $n\ge 10$ and $m\ge 62$. For example, taking $n=16$ and $m=64$ we get a channel with $4$ bits of input, $1$ bit of output and $6$ bits of memory.

Indeed, we will prove something slightly stronger. Let us recall that a decision problem can be cast as a function with values in $\{0,1\}$, where $1$ stands for accept and $0$ for reject. When the associated function is uncomputable, the decision problem is called undecidable. 

Fix a rational number $\lambda\in (0,1]$. We will give explicitly a subfamily  $\mathcal{S}_{\lambda}$ of FSMCs (information stable and with rational conditional probability assignments in product form) with $10$ input symbols and $62$ states, with the additional property that all channels $\mathbf{N}\in\mathcal{S}_{\lambda}$ have capacity $\ge \lambda$ or $\le \lambda/2$.  

\begin{main}\label{main2}
It is undecidable to know whether  $\mathbf{N}\in \mathcal{S}_{\lambda}$, given by its set of probability assignments $\{p(y|x,s'), p(s|x,s')\}_{s,y,x,s'}$, has capacity $\ge \lambda$ or $\mathbf{N}\le \lambda/2$.
\end{main}

It is clear that if we consider our Main Result \ref{main2} for $\mathcal{S}_{1}$ we get Main Result \ref{main1}. 
That is, if we could approximate the capacity within error $1/5$ then, given a channel from $\mathcal{S}_1$ for which we know its capacity is $\le 1/2$ or $\ge 1$, we could decide which is the case. However, we know by Main Result \ref{main2} that the problem is undecidable.

\subsection*{Proof sketch}
\label{sec:proofsketch}

\

The idea behind our proof is to construct a family of channels such that the capacity of a channel in the family is related to some property of a probabilistic finite automaton (PFA).
Our construction is indirect, we first give a map form PFAs to FSMCs; then we define the channel family as the set of FSMCs that are the image of a PFA via this map.
The important property of this map, proved in Theorem \ref{maintheorem2}, is that the capacity of a channel in the image set is given by the value of its preimage PFA (see the PFA section). 
We now sketch the structure of the proof and point to the appropriate sections for further detail.

FSMCs, defined in \ref{sec:fsmc}, are controlled by a finite state machine. The state of the finite state machine determines the (memoryless) channel that is applied to the input. Then depending both on the input and the current state it transitions probabilistically to the next state. A PFA is a finite state machine that transitions probabilistically from state to state depending on the current state and the input (see Figure \ref{fig:exrubik}). Hence, it is possible to identify the finite state machine controlling a FSMC with a PFA.

A concrete input into an PFA, that is a sequence of input symbols, is accepted if after reading the input the PFA ends in a subset of the states called accepting states, otherwise the input is rejected.  Informally, the value of a PFA is the maximum probability of ending in an accepting state. It turns out that many decision problems related to the value can not be solved. Notably, given some value $\lambda\in(0,1)$ and a PFA $\A$ it is undecidable to know if the value of $\A$ is greater than $\lambda$ \cite{Ginsburg_66,Paz_71,Condon_89}. Here, we use a recent proof of this result by Hirvensalo \cite{Hirvensalo_07}, see Theorem \ref{th:emproblem} in  \ref{subsec:hirvensalo}.  
In order to prove a result about approximations, we amplify this result about decision problems with a very original PFA construction by Gimbert and Oualhadj \cite{Gimbert_10}, which is the key ingredient in our proof, see Lemma \ref{lem:construction} in  \ref{subsec:gimbert}. With this construction, it is possible to embed any PFA $\A$ into a larger PFA $\B_\lambda$ (with $\lambda\in[0,1/2]$) such that: the value of $\B_\lambda$ is $\leq \lambda$ if and only if the value of $\A$ is $\leq 1/2$ and the value of $\B_\lambda$ is $2\lambda$ if and only if the value of $\A$ is $>1/2$. Joining both arguments, we conclude that the value of a PFA can not be approximated with arbitrary precision, since it is undecidable to know whether the value of a PFA is smaller than $\lambda$ or equal to $2\lambda$.

To go from there to Main Result \ref{main2} it is enough to construct for any PFA $\A$ a channel $\mathbf V_\A$  so that the capacity of $\mathbf V_\A$ equals the value of $\A$. The idea for that is very natural:

\

Consider a channel with two input registers. The first one is used to control the PFA. The second one corresponds to the data to be transmitted. If the PFA is in an accepting state the channel outputs the contents of the second input register, that is, it behaves as a noiseless channel. Otherwise it is only noise, i.e. it outputs uniformly at random a symbol from the output alphabet.

\

Intuitively, this map should already have the property that the capacity of $\mathbf V_\A$ equals the value of $\A$. However, without an additional gadget, we cannot conclude this. Let us see with an example why it does not suffice. Consider for instance a PFA that transitions from the initial state to an accepting state with probability $1/2$ and with probability $1/2$ to some other state. Moreover, suppose that these two states are final in the sense that the PFA can not leave them once reached. Such a PFA would have value $1/2$. However, the capacity of the associated channel would be zero because the error probability of any code would always be greater than $1/4$. In order to solve this problem we concatenate the map with a function $\gamma(\cdot)$ from PFAs to PFAs that adds to the PFA a reset and a freeze symbols. The reset symbol takes the PFA back to the initial state while the freeze symbol keeps the state of the PFA unchanged. We prove in Lemma \ref{lem:reslemma} in  \ref{subsec:resandstab}, that the value of an automaton $\A$ does not change under this map, i.e. the value of $\A$ equals the value of $\gamma(\A)$. But, for PFAs with the additional reset and freeze symbols we can show the desired result that the capacity of the channel $\mathbf V_{\gamma(\A)}$ equals the value of the automaton $\A$. This is our main technical result, proved in Theorem \ref{maintheorem2}.

The intuition between the equality of the capacity of channel $\mathbf V_{{\mathcal{A}}}$ and the value of $\mathcal A$, $\valA$, is as follows. 
For any $\delta>0$ there exists a word 
with value greater than $\valA-\delta$. 
By feeding this word into the control register, the channel will transition into a final state with probability at least $\valA-\delta$. 
The state of the channel can then be frozen making the mutual information rate tend to $\valA-\delta$.
However, this rate might not be achievable. In order to show achievability then, we induce a memoryless channel by choosing for the control input a periodic sequence that ends with a reset symbol.
More concretely, for $\delta>0$, the sequence consists of: a word with a value larger than $\valA-\delta$, a number of freeze symbols that guarantee an information rate larger than $\valA-2\delta$ and a reset symbol.
In the other direction, one would not expect a capacity larger than $\valA$. The reason is that the channel outputs a symbol uniformly at random when it is in a non-final state and this happens with probability at least $1-\valA$.

Finally, note that at this point we do not know yet that the channel is information stable. Indeed, the proof of this fact (Corollary \ref{cor:infest} in  \ref{Appendix:inf-stable}) will use crucially Theorem \ref{maintheorem2}.

\subsection*{Formal statements of the main results}\label{sec:tm}

\

So far we have introduced the notion of uncomputable functions as those that can not be computed with an algorithm (similarly the notion of undecidable problems). In order to make this definition, and hence the Main Results, mathematically rigorous, we have to recall the definition of a Turing Machine (TM) as the formal definition of what an {\it algorithm} is. For more details one can consult for instance \cite{Sipser_06,Arora_09}.

A TM represents a machine with a finite set of states that can read from and write to an infinitely long memory in the form of a tape. The tape is divided into cells that can hold a single symbol from a finite alphabet. Initially, the tape contains some arbitrary but finite string that we call the input followed by an infinite sequence of blank symbols. 
The operation of the machine is controlled by a head that sits on top of a cell of the tape. The head operates as follows: it reads the symbol below it; then, depending on the symbol and the current state it writes a symbol, moves left or right and transitions to a new state. The set of states includes the halting state. 
The TM halts after it transitions to the halting state. 
The output of the TM consists of the, possibly empty, string of symbols starting from the leftmost non-blank symbol to the rightmost non-blank symbol. 

Formally, a TM is defined by a triple $M=(Q,\Sigma,\delta)$ where $Q$ represents the finite set of states including an initial and a halting state, $\Sigma$ is the finite set of symbols that a cell may contain and it includes the blank symbol and $\delta:(Q\times\Sigma)\mapsto(Q\times\Sigma\times\{L,R\})$ is the transition function.  

A configuration is a complete description of the status of a TM. It consists of 
the current state, 
the contents of the tape and 
the position of the head. In the initial configuration, the tape contains the input string and the head of the TM is in the initial state and situated on top of the leftmost cell of the input. Once the initial configuration is fixed a TM evolves deterministically and may or may not eventually halt. 

Let us fix $n=10$, $m= 62$. In order to specify a FSMC with $n$ input symbols and $m$ states, it is enough to give $N=nm(2+m)=39680$ rational numbers corresponding to the conditional probability assignments. 
It is very easy to construct an injective map $\sigma(\cdot)$ from vectors of $N$ positive rational numbers to the natural numbers (see  \ref{sec:pfaencod}), which then can be transformed into a valid input of a TM. For instance, it would be transformed into a string of zeroes and ones if $\Sigma=\{0,1,\#\}$. 
Main Results  \ref{main1} and \ref{main2} can be then respectively restated as:

\begin{main1*}
There does not exist any TM that  halts on all inputs of the form $\sigma(\mathbf{N})$ for $\mathbf{N}\in\mathcal S_1$ and outputs a rational number $c$ such that the capacity of $\mathbf{N}$ verifies
\begin{equation}
\left|C(\mathbf{N})-c\right|\le \frac{1}{5}\; .
\end{equation}
\end{main1*}

\begin{main2*}
There does not exist any TM that halts on all inputs of the form $\sigma(\mathbf{N})$ for $\mathbf{N}\in \mathcal{S}_{\lambda}$ and outputs $1$ if the capacity of $\mathbf{N}\ge \lambda$ and $0$ if the capacity of $\mathbf{N}\le \lambda/2$.
\end{main2*}

\begin{figure*}
\includegraphics[width=175mm]{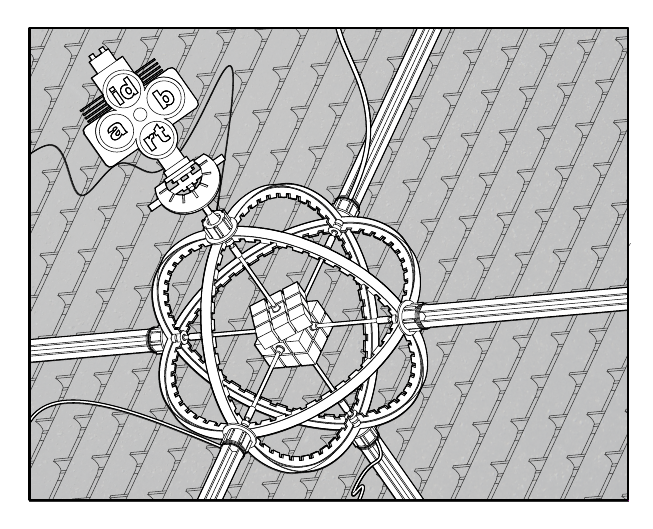}
\caption{A noisy Rubik cube solver as an example of a Probabilistic Finite Automaton (PFA)\\ This PFA has as many states as different Rubik cube configurations. It begins in some predefined state and can be manipulated with four different buttons or input alphabet symbols: \{$a,b,\textrm{id},\textrm{rt}$\}. A Rubik cube can be solved by combinations of only two sequences of rotations \cite{singmaster1981notes}. The press of the buttons $a,b$ will, with some probability, implement one of these two sequences and with the complementary probability apply a random rotation. The buttons $\textrm{id},\textrm{rt}$ will make the state of the Rubik cube either stay idle or bring it back to the initial state. The accepting state is the solved configuration of the cube. The value of this automaton would be the maximum probability of taking the initial configuration to the solved configuration by pressing a sequence of buttons. (Credit: Francisco Garc\'ia Moro)}\label{fig:exrubik}
\end{figure*}

\
\section*{Discussion}
\label{sec:disc}

We have proven that no algorithm can exist that approximates the capacity for all information stable 
FSMCs to any desired precision. 
 
Our construction builds directly on top of several strong undecidability results of PFAs. Recent developments underlying these results suggest that it should be possible to reduce the dimensions of our construction \cite{Neary_15}. It is an interesting problem to find the minimal dimensions for which uncomputability holds.

It is important to notice also that the channels appearing in our construction have long term memory. Combined with the known results for memoryless channels, this suggests the existence of a tradeoff between the time-scale of the memory of a channel and the efficiency to compute its capacity. Giving precise quantitative bounds in this direction is an interesting open question.

It is also worth exploring other problems that could be attacked with similar techniques. The proof technique  
can be extended to the capacities of quantum channels with memory implying an even stronger inapproximability result in that case. We will make the explicit analysis in a forthcoming paper. Similar long term memory effects appear in other interesting situations, associated with other entropic quantities. One paradigmatic example is cryptography, where in order to analyze the security of the sequential use of a device, one needs to assume the worst case-scenario in which the adversary keeps memory of its past uses. Both in the classical and in the quantum case, the techniques of this paper could provide insights on the difficulty to provide optimal results in cryptographic settings. 

Furthermore, our result connects with recent work regarding the different capacities of memoryless quantum channels \cite{Cubitt_15,Elkouss_15}, showing some evidence that these capacities might be uncomputable. Also, memoryless zero error capacities, both classical and quantum, are known to have highly non-trivial behaviour \cite{Alon_98,Chen_11,Cubitt_11,Cubitt_12,Shirokov_14}. Unfortunately, the techniques used here exploit directly the memory of the channel and hence cannot be directly applied to the memoryless capacities. The question is however of unquestionable interest.

\
\section*{Methods}
\subsection*{Notation}
\label{sec:notation}

We denote random variables by capital letters $X,Y,...$, sets and probabilistic finite automata (PFA) -see below for the definition- by calligraphic capital letters $\mathcal X, \mathcal Y, ...$, channels by capital bold face letters $\mathbf X, \mathbf Y, ...$, and instances of random variables by lower case letters $x, y, ...$. We denote vectors with the same convention, whenever confusion might arise a superscript indicates the number of components of the vector and a subscript the concrete component: $X^n=(X_1,X_2,...,X_n)$ or $x^n=(x_1,x_2,...,x_n)$. We indicate a consecutive subset of $n$ components of the vector with subscript notation $[a,a+n-1]$: $x_{[a,a+n-1]}=(x_a,x_{a+1},...,x_{a+n-2},x_{a+n-1})$. 

A vector is called a probability vector if all its entries are non-negative and add up to one. 
A matrix is called a stochastic matrix if all its columns are probability vectors. A stochastic matrix takes probability vectors to probability vectors.

\subsection*{Probabilistic finite automata}
\label{sec:pfa}

\

A PFA consists of a finite set of inputs and a finite set of states. One of these states is the initial state and a subset of the states are accepting states. 

The action of the PFA is defined by the transition probabilities from one state to another as a function of the input symbols. A word is a sequence of symbols. After a word is fed to a PFA in the initial state, the PFA will transition from state to state and will end up in an accepting state with some probability. We call this probability the accepting probability of a word. Intuitively, we can understand a PFA as a machine with noisy knobs, the input symbols, and the input word is a sequence of knobs that tries to steer the machine into some desired state. See Figure \ref{fig:exrubik} for an example and \ref{sec:pfaap} for a formal definition.

Given some PFA, we denote by $\valA$ the supremum of the acceptance probabilities over all input words:
\begin{equation}
\valA=\sup_{\word{w}}
\val(\A,\word{w})
\end{equation}
where $\val(\A,\word{w})$ denotes the value of $\word{w}$ when input into the PFA $\A$ and the optimization runs over all words of finite length.

We consider two types of PFAs that we name as {\it \stable}\ and {\it resettable}. 

We call a PFA a \stable\ PFA if one of the transition matrices is equal to the identity matrix $\mathcal X_{\id}$. The reason is that for such a PFA reading the symbol corresponding to the identity leaves the state probabilities unchanged. Let $u$ be any probability vector, then
\begin{align}
u &= \mathcal X_{\id} u 
\label{eq:id}             
\end{align}

We call a PFA a resettable PFA if one of the transition matrices, $\mathcal X_\rt$, takes the state back to the initial state. Let $u$ be any probability vector, then
\begin{align}
v &= \mathcal X_\rt u 
\label{eq:reset}             
\end{align}

We let $\gamma$ be a map from PFAs to PFAs such that for all PFA $\A$, $\gamma(\A)$ is \stable\ and resettable. More concretely:
\begin{definition}\label{def-gamma}
Given a PFA $\A=\{\Q,\W,\X,v,\F\}$, we define $\gamma(\A)=\{\Q,\W\cup\{\id,\rt\},\X\cup\{\X_\id,\X_\rt\},v,\F\}$ as an automaton that extends $\A$ with the two additional input symbols $\id$ and $\rt$ and the corresponding matrices $\mathcal X_\id$ and $\mathcal X_\rt$ as given by \eqref{eq:id} and \eqref{eq:reset}. 
\end{definition}

The key lemma we will need about PFAs, essentially due to Gimbert and Oualhadj \cite{Gimbert_10}, is the fact that their value cannot be approximated within a constant error. Let us give the precise statement. The complete proof can be found in  \ref{proofs}. Fix a rational number $\lambda\in (0,1]$. 
\begin{lemma}\label{mainlemma1}
One can give explicitly a subfamily $\mathcal{T}_{\lambda}$ of rational \stable\ and resettable PFA with alphabet size $5$ and $62$ states with the following properties:
\begin{enumerate}
\item[(i)] $\valA$ is either $\ge \lambda$ or $\le \lambda/2$ for all $\mathcal{A}\in \mathcal{T}_{\lambda}$.
\item[(ii)] It is undecidable to know which is the case.
\end{enumerate}
\end{lemma}

The definition of $\mathcal{T}_{\lambda}$ will be given in  \ref{proofs}, Eq. (\ref{def-family-T}). 

\subsection*{The family $\mathcal{S}_{\lambda}$ and the proof of Main Result \ref{main2}}
\label{sec:pfacaps}

\

Given a \stable\ and resettable PFA $\A$ we define the channel $\mathbf V_{{\mathcal{A}}}$ as follows. The input alphabet of the channel takes values in $\{0,1\}\times\mathcal W$, which we identify with two different input registers: a data input and a control input. 
The data input is transmitted to the output: noiselessly if $\A$ is in an accepting state or, if $\A$ in any other state, the channel outputs uniformly at random an element of the output alphabet. More concretely, the output of the channel is defined by the following conditional probability:
\begin{align}
\label{eq:pfatoch1}
p(y_n|x_n,s_{n-1}) &=\left\{
        \begin{aligned}      
        &\left. \frac{1}{2}\right.  &\textrm{if }s_{n-1}\notin\F\\
        &\left. 1\right. &\textrm{if }s_{n-1}\in\F\textrm{ and }y_n=x_n\\
        &\left. 0\right. &\textrm{else}
       \end{aligned}
\right.
\end{align}

The control input is fed to $\A$, which begins in the initial state, and the state transition probabilities are dictated by the PFA:
\begin{equation}
p(s_n|c_n,s_{n-1}) = \left\langle\pi_{s_n}, \mathcal X_{c_n} \pi_{s_{n-1}}\right\rangle\label{eq:pfastate}\ .
\end{equation}

We connect the properties of PFA with the capacity of FSMCs in the next Theorem:
\begin{theorem}\label{maintheorem2}
The capacity of $\mathbf V_{{\mathcal{A}}}$ is given by:
\begin{align}
C(\mathbf V_{{\mathcal{A}}}) &= \valA
\end{align}
\end{theorem}
We defer the proof to  \ref{sec:proof-lemma2}. 

The family $\mathcal{S}_{\lambda}$ in Main Result \ref{main2} is defined simply as
$$\mathcal{S}_{\lambda}=\{\mathbf V_{{\mathcal{A}}}: \mathcal{A}\in \mathcal{T}_{\lambda}\}\;.$$
with $\mathcal{T}_{\lambda}$ the family introduced in Lemma \ref{mainlemma1} and defined in  \ref{proofs}, Eq. (\ref{def-family-T}).

Main Result \ref{main2} is then a trivial consequence of Lemma \ref{mainlemma1} and Theorem \ref{maintheorem2}.

Furthermore one can leverage Theorem \ref{maintheorem2} to show that all the channels in $\mathcal{S}_{\lambda}$ are information stable:
\begin{corollary}
\label{cor:infest}
Given $\mathbf V_{\mathcal A}\in\mathcal S_{\lambda}$, $\mathbf V_{\mathcal A}$ is information stable.
\end{corollary}

The proof will be given in  \ref{Appendix:inf-stable}.

\

Data sharing not applicable to this article as no datasets were generated or analysed during the current study.

\section*{Acknowledgements}
The authors are grateful to Toby Cubitt, Jeremy Ribeiro, Eddie Schoute, Henry Pfister, Marco Tomamichel and Michael Wolf as well as to the anonymous referees for stimulating discussions and feedback. We thank Kenneth Goodenough for pointing us to the upper bound on the Riemann zeta function and David Reeb and Tobias Koch for extremely valuable comments that highly improve the presentation of the paper.

DPG acknowledges support from MINECO (grant MTM2014-54240-P), from Comunidad de Madrid (grant QUITEMAD+- CM, ref. S2013/ICE-2801), and the European Research Council (ERC) under the European Union's Horizon 2020 research and innovation programme (grant agreement No 648913). This work has been partially supported by ICMAT Severo Ochoa project SEV-2015-0554 (MINECO). DE acknowledges support from TU Delft Open Access Fund. 



\newpage
\renewcommand{\thesection}{Supplementary Note \arabic{section}}
\renewcommand{\figurename}{Supplementary Figure}
\renewcommand{\bibsection}{\section*{Supplementary References}}
\widetext
\section{PFA}
\label{sec:pfaap}
A PFA ${\A}$ is given by a tuple $\A=(\Q, \W, \X, v, \F)$. $\Q$ denotes a finite set of states, $\W$ denotes a finite input alphabet, $\X$ denotes a finite set of stochastic matrices with cardinality equal to the cardinality of the input alphabet, $v$ denotes an initial probability distribution over $\Q$ and $\F\subseteq \Q$ denotes a set of accepting states. We say that the PFA is {\it rational} if the coefficients of $\X$ and $v$ are rational numbers. We will only consider rational PFA in the sequel.

The action of a PFA is defined by the transition probabilities from one state to another as a function of the input symbols. If the automaton is in the state $q_a$ and reads the letter $w$ it transitions to the state $q_b$ with probability:
\begin{align}
p\left[q_a\xrightarrow{\,w\,}q_b\right] &= (X_w)_{q_b,q_a} \\
                                                          &= \left\langle \pi_{\{q_b\}}, X_w \pi_{\{q_a\}}\right\rangle
\end{align}
where we denote by $\langle a,b\rangle$ the scalar product between vectors $a$ and $b$ and by $\pi_\X$ a vector with ones in the positions indicated by $\X$ and zeroes in the remaining positions. 

We exploit the same notation for the probability that the automaton transitions from the state $q_a$ to the state $q_b$ after reading the word $\word{w}=(w_1,\ldots,w_{|\word{w}|})\in\W^{|\word{w}|}$: 
\begin{equation}
p\left[q_a\xrightarrow{\,\word{w}\,}q_b\right] = \left\langle \pi_{\{q_b\}},  X_{w_{|\word{w}|}}\cdot\ldots\cdot X_{w_1} \pi_{\{q_a\}}\right\rangle\ .
\end{equation}
More generally, if we have a probability distribution over the states given by the column vector $x$ and the PFA reads the letter $w$ then the new distribution over the states is given by $X_w x$. A particularly relevant probability is the probability that the automaton ends in an accepting state after reading some word $\word{w}$. 
We call this probability the probability of accepting $\word{w}$ or the value of $\word{w}$. It can be computed 
\begin{equation}
\val(\A,\word{w})=\left\langle\pi_{\F}, X_{w_{|\word{w}|}}\cdot\ldots\cdot X_{w_1} v\right\rangle \ .
\end{equation}
We call the value of $\A$, which we denote by  $\valA$, 
the supremum 
of the acceptance probabilities over all input words:
\begin{equation}
\valA=\sup_{\word{w}\in\mathcal W^*} \val(\A,\word{w})
\end{equation}
where we denote by $\mathcal W^*$ the set of finite length words in $\mathcal W$.

Whenever possible, we will represent graphically the different automata constructions. We will follow the following conventions. A state is denoted by a circle. An accepting state is denoted by a circle with a double line around it. In all automata in this paper, the initial distribution will have one coefficient with weight one. We indicate the corresponding state with an arrow that does not come from any state. 

We indicate with $\xrightarrow{w,p}$ that if the automaton reads the letter $w$ it transitions from the origin of the arrow to the state pointed by the arrow with probability $p$. In order to avoid clutter, we simplify the notation in several cases.
If we do not show transitions corresponding to all input symbols, the missing transitions correspond to self-loops with probability one. We drop the probability and just write $\xrightarrow{w}$ if a transition occurs with probability one. We drop the input symbol and just write $\xrightarrow{p}$ if all input symbols transition with the same probability. 

\begin{example} Consider the PFA given in \figurename~\ref{fig:exautomata}. 
The automaton in the figure has three states $\Q=\{q_1,q_2,q_3\}$, two input symbols $\W=\{a,b\}$, the initial state is $q_1$ and there is a single accepting state $q_3$. By looking at the figure we can construct the stochastic matrices:
\begin{equation}
X_a=\begin{pmatrix}
  0.5 & 1 & 0 \\
  0.5 & 0 & 0.5 \\
  0    & 0 & 0.5 
 \end{pmatrix},\quad X_b=\begin{pmatrix}
  0 & 0 & 0 \\
  0 & 1 & 0.5 \\
  1 & 0 & 0.5 
 \end{pmatrix}
\end{equation}
Now, assume that we see the word $\word{w}=baa$, we can easily compute its value:
\begin{equation}
\begin{pmatrix}0&0&1\end{pmatrix}X_b\cdot X_a\cdot X_a\begin{pmatrix}1\\0\\0\end{pmatrix}=0.25
\end{equation}
\label{ex:pfa}
\begin{figure}
\centering
{\includegraphics[width=7.5cm]{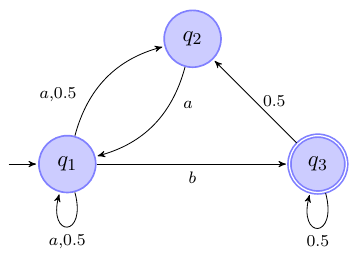}
}
\caption{Automaton with three states $\Q=\{q_1,q_2,q_3\}$, two input symbols $\W=\{a,b\}$, the initial state is $q_1$ and there is a single accepting state $q_3$. 
}\label{fig:exautomata}
\end{figure}
\end{example}

\section{Finite State Machine Channels}\label{sec:fsmc}

\

A channel can depend on past inputs and outcomes in very complicated ways. 
We focus our interest on finite FSMC which is the set of discrete channels that have its behavior dictated by a finite state machine \cite{Gallager_68}. 
Let $\mathcal X$, $\mathcal Y$ and $\mathcal S$ be finite sets that represent the input alphabet, output alphabet and set of states. A FSMC is characterized by the time-invariant conditional probabilities $p(y,s|x,s')$  for all states $s,s' \in\mathcal S$, input symbols $x\in\mathcal X$ and output symbols $y\in\mathcal Y$. These conditional probabilities denote the probability that the channel outputs the symbol $y$ and transitions to the state $s$ given that the channel is in state $s'$ and receives the input symbol $x$. In this paper, we will restrict our attention to those FSMC for which $p(y,s|x,s')$ have a product form $p(y|x,s')p(s|x,s')$.

We assume that the initial state $s_0$ is known by both the transmitter and the receiver. We denote by $W^n_{s_0}$ the sequence of probability distributions induced by $s_0$ that give the probability of a sequence of outputs given a sequence of inputs into the channel:
\begin{align}
W^n_{s_0}\left(y^n|x^n\right)&=\sum_{s_{n}}W^{n}_{s_0}\left(y^{n}s_n|x^{n}\right)
\end{align}
where
\begin{align}
W^n_{s_0}&\left(y^ns_n  |x^n\right)=\nonumber\\
=&\sum_{s_{n-1}}p(y_n,s_n|x_n,s_{n-1})W^{n-1}_{s_0}\left(y_{n-1}s_{n-1}|x_{n-1}\right)
\end{align}
Analogously we can define a sequence of probability distributions to characterize the state of the channel:
\begin{align}
W^n_{s_0}\left(s_n|x^n\right)&=\sum_{y^{n}}W^{n}_{s_0}\left(y^{n}s_n|x^{n}\right)
\end{align}
Note that we have used, abusing the notation, $W^n_{s_0}$ both to define the conditional probability of the output and the state. Consider two random variables $X,Y$ with joint distribution $p(x,y)$, the information spectrum is the distribution of the random variable $i_{X,Y}$ given by:
\begin{equation}
i_{X,Y}(x,y)=\log \frac{p(y|x)}{p(y)}
\end{equation}
The mutual information is the expected value of the information spectrum:
\begin{align}
I(X;Y) &= \langle i_{X,Y}(X,Y)\rangle_{XY}\\
          &= \sum_{x,y} p(x,y) \log \frac{p(y|x)}{p(y)}
\end{align}
A channel is said to be information stable\ \cite{Dobrushin_63} if for all $\gamma>0$ there exists a sequence of random variables $\{X_i\}_{i=1}^\infty$ such that:
\begin{align}
\lim_{n\rightarrow\infty}\textrm{Pr}\left[\left|\frac{i_{X^n;W^n}\left(X^n;Y^n\right)}{nC^n}-1\right|>\gamma\right]=0
\end{align}
where:
\begin{align}
C_n=\sup_{X^n}\frac{1}{n} I(X^n;Y^n)
\end{align}
In full generality one may need to resort to the capacity formula of Verdu and Han \cite{Verdu_94} in order to compute the capacity of a FSMC. However, 
if the channel is information stable then the capacity is given by the mutual information rate \cite{Dobrushin_63}:
 \begin{align}
	 C(\mathbf W) &=\lim_{n\rightarrow\infty}\sup_{X^n} \frac{1}{n}I(X^n;Y^n)
 \end{align}

\section{An encoding for FSMC into Turing machines}
\label{sec:pfaencod}
Let us construct an explicit example of the $\sigma$ function introduced in the main text.

Since the set of FSMC can be seen as a subset of positive elements in $\mathbb{Q}^{N}$, it is enough to give an explicit injective map $\sigma$ from the set of positive elements in $\mathbb{Q}^{N}$ into the natural numbers $\mathbb{N}$. For instance, consider the first $2N$ prime numbers $p_1,\ldots, p_{2N}$ and define $$\sigma\left(\frac{r_1}{s_1},\ldots, \frac{r_{N}}{s_{N}}\right)= \prod_{j=1}^{N} p_j^{r_j}\prod_{j=N+1}^{2N} p_j^{s_j}\;.$$ Any other explicit $\sigma$ will do the job. 

\section{Proof of Theorem \ref{maintheorem2}}\label{sec:proof-lemma2}

\begin{proof} 
First we will prove that $\valA$ is an achievable rate, that is: $C(\mathbf V_{{\mathcal{A}}})\geq\valA$. Then we will prove that $\valA$ is an upper bound on the mutual information rate and in consequence $C(\mathbf V_{{\mathcal{A}}})\leq\valA$.

Let $\delta>0$, then there exists some word $\word{w}$ such that $\val(\A,\word{w})\geq\valA-\delta$, furthermore let $|\word{w}|=m$. Consider the following protocol, the input into the control register is the deterministic sequence $(c_i)_{i=1}^\infty$ with 

\begin{align}
c_i &=\left\{ 
\begin{aligned}
        &\left. w_{i} \right. &\textrm{if } i-1 \mod m+n < m \\
        &\left. \rt \right. &\textrm{if } i-1 \mod m+n = m+n-1 \\
        &\left. \id \right. &\textrm{else}
\end{aligned}
\right.        
\label{eq:control}
\end{align} 

This choice induces a memoryless channel when regarded in blocks of $m+n$ uses of the channel. That is, every block of $m+n$ inputs into the data input encounters exactly the same noisy channel once the control input is fixed by Supplementary Equation \eqref{eq:control}. In consequence, given this particular control input, any mutual information between the input and the output over $m+n$ uses is an achievable rate (once normalized over the number of uses). For the data input, we choose the uniform distribution. 

The following chain of inequalities holds for the conditional entropy of the output given the input:
\begin{align}
&H\left(Y_{[1,{m+n}]}|X_{[1,{m+n}]}C_{[1,{m+n}]}\right)\nonumber\\
                        &=\sum_{i=1}^{m+n} H\left(Y_i|Y_{[1,i-1]}X_{[1,{m+n}]}C_{[1,{m+n}]}\right)\label{eq:chainrule}\\
                        &\leq m+H\left(Y_{[m+1,{m+n}]}|Y_{[1,m]}X_{[1,{m+n}]}C_{[1,{m+n}]}\right)\label{eq:boundbym}\\
                        &\leq m+H\left(Y_{[m+1,{m+n}]}|X_{[1,{m+n}]}C_{[1,{m+n}]}\right)\label{eq:removecond}\\
                        &\leq m + 1+ (1-\valA+\delta)n\label{eq:lasteq}
\end{align}
where Supplementary Equation \eqref{eq:chainrule} follows by the chain rule, the inequality Supplementary Equation \eqref{eq:boundbym} by bounding the entropy of the first $m$ uses by $m$, the inequality Supplementary Equation \eqref{eq:removecond} by removing the conditioning on $Y_{[1,m]}$ and Supplementary Equation \eqref{eq:lasteq} holds from bounding the conditional entropy by Supplementary Equation \eqref{eq:hym1n} that we prove below.

After the first $m$ uses, the automaton 
behaves like a noiseless channel with probability at least $\valA-\delta$ and like a completely random channel with the complementary probability. In consequence, we can bound the conditional entropy of the output of the uses $m+1$ to $m+n$ as follows:
\begin{align}
H\big(&Y_{[m+1,{m+n}]}|X_{[1,{m+n}]}C_{[1,{m+n}]}\big)\\                   
&\leq H\big((\valA-\delta+(1-\valA+\delta)2^{-n}) \phi\nonumber\\
                        &\qquad\qquad +(1-\valA+\delta)(1-2^{-n})\rho\big) \label{eq:iineqfirst} \\
                        &= h(\valA-\delta+(1-\valA+\delta)2^{-n})\nonumber\\
                        &\qquad\qquad+(1-\valA+\delta)(1-2^{-n})\log(2^{n}-1)\\
                        &\leq 1+ (1-\valA+\delta)n \label{eq:hym1n}
\end{align}
where $\phi=(1,0,\ldots,0)$ is a completely deterministic probability vector of length $2^{n}$, $\rho=\left(0,\frac{1}{2^{n}-1},\frac{1}{2^{n}-1},\ldots,\frac{1}{2^{n}-1}\right)$ the maximally entropic vector of length $2^{n}-1$ and $h(\epsilon):= - \epsilon \log \epsilon - (1-\epsilon) \log (1-\epsilon)$ is the binary entropy function. 

Now we can use Supplementary Equation \eqref{eq:lasteq} to bound the mutual information of the first $m+n$ uses:
\begin{align}
I(Y^{m+n}&;X^{m+n}C^{m+n})\\
 &=H(Y^{m+n})-H(Y^{m+n}|X^{m+n}C^{m+n})\\
 &=m+n-H(Y^{m+n}|X^{m+n}C^{m+n})\\
 &\geq n(\valA-\delta)-1 \label{eq:iineqlast}
\end{align}
Finally, by choosing $n$ larger than $(1+(\valA-2\delta)m)/\delta$ we get 
\begin{equation}
\frac{1}{m+n}I(Y^{m+n};X^{m+n}C^{m+n})\geq \valA-2\delta
\end{equation}
 That is, for all $\delta>0$ the rate $\valA-2\delta$ is achievable.

Now we will prove that $C(\mathbf V_{\mathcal A})$ is upper bounded by $\valA$.

Let $S_i$ denote the state of the PFA at use $i$, since the output only depends on the control input through the PFA state we have that $H(Y_i|X_iC_{[1,i-1]})\geq H(Y_i|X_iS_{i-1})$. In consequence, we can bound from below the conditional entropy of the output given the input as follows:
\begin{align}
H\big(Y^{n}&|X^{n}C^{n}\big) \nonumber\\
        &= \sum_{i=1}^n H\left(Y_i|Y_{[1,i-1]}X^{n}C^{n}\right)\label{eq:ineqfirst} \\
                                  &= \sum_{i=1}^n H\left(Y_i|Y_{[1,i-1]}X_iC_{[1,i-2]}\right) \\
                                  &\geq \sum_{i=1}^n H(Y_i|S_{i-1}X_i) \\
                                  &=  \sum_{i=1}^n p(S_{i-1}\in\F)H(Y_i|S_{i-1}\in\F,X_i=0)\nonumber\\
                                  &\qquad+p(S_{i-1}\notin\F)H(Y_i|S_{i-1}\notin\F,X_i=0)\\
                                  &\geq n(1-\valA)\label{eq:ineqfinal}
\end{align}

Finally, we can plug the bound on the conditional entropy to obtain the desired result:
\begin{align}
\frac{1}{n}I(Y^{n};X^{n}C^{n}) &= \frac{1}{n}(H(Y^n) - H(Y^{n}|X^{n}C^{n})) \\
                                  &\leq  1 -  \frac{1}{n}H(Y^{n}|X^{n}C^{n}) \\
                                  &\leq \valA
\end{align}

\end{proof}

\section{Proof of Corollary \ref{cor:infest}} \label{Appendix:inf-stable}
\begin{proof}
From the proof of Theorem \ref{maintheorem2} we know that $\forall \delta>0$, $\forall t\in\mathbb N$ there exists $n_t$ (wlog $n_{t+1}\ge n_t$) and $X^t=\{X_1,\ldots,X_{n_t}\}$ such that:
\begin{equation}
\frac{I(X^t;Y^t)}{n_t}\geq \valA -\frac{\delta}{2^t}
\end{equation}
We define the following source to input into the channel:
\begin{equation}
\mathbf V =\left\{\underbrace{X^1,\ldots,X^1}_{m_1\textrm{ times}},\ldots,\underbrace{X^t,\ldots,X^t}_{m_t\textrm{ times}},\ldots\right\}
\end{equation}
Each use of the channel is uniquely identified by a triple $(t,\alpha,\beta)$ with $t\in\mathbb N$, $\alpha\in[0,m_{t+1}-1]$ and $\beta\in[0,n_{t+1}]$ such that the triple corresponds with the use $n$-th with
\begin{equation}
\label{eq:tnrelation}
n=\sum_{i=1}^tm_in_i+\alpha n_{t+1}+\beta 
\end{equation}
and the sequence of random variables that is input over the first $n$ uses is 
\begin{align}
&V^n=\\&=\left\{\underbrace{X^1,\ldots,X^1}_{m_1\textrm{ times}},\ldots,\underbrace{X^t,\ldots,X^t}_{m_t\textrm{ times}},\underbrace{X^{t+1},\ldots,X^{t+1}}_{\alpha\textrm{ times}},X^{t+1}_{[1,\beta]}\right\}\nonumber
\end{align}
The sequence $\{m_i\}_{i=1}^\infty$ is chosen such that: 
\begin{equation}
\label{eq:boundiVn}
\frac{I\left(V^n;W^n\right)}{n}\geq \valA -\frac{\delta}{2^{t-1}}
\end{equation}
where $W^n$ is the random variable induced by $V^n$ at the output of the channel and $n$ is related to $t$ by Supplementary Equation \eqref{eq:tnrelation}. 

Let $C_n=\sup_{X^n}I(X^n;Y^n)/n$, Supplementary Equation \eqref{eq:boundiVn} implies that
\begin{equation}
	1\geq\mathbb E\left[\frac{i_{V^n,W^n}}{nC_n}\right]\geq\mathbb E\left[\frac{i_{V^n,W^n}}{n\valA}\right]\geq 1-\frac{\delta}{\valA2^{t-1}}
\end{equation}
Note that the input $V^n$ is composed of independent random variables, hence:
\begin{equation}\label{eq:iid-information-spectrum}
i_{V^n,W^n}=\sum_{i=1}^{t}\sum_{j=1}^{m_i}i_{X^i,Y^i}+\sum_{i=1}^{\alpha}i_{X^{t+1},Y^{t+1}}+i_{X^{t+1}_{[1,\beta]},Y^{t+1}_{[1,\beta]}}
\end{equation}
and also note that for all $t$ and $\beta\in[0,n_t]$:
\begin{equation}
\label{eq:boundsiXt}
0\leq i_{X^t_{[1,\beta]},Y^t_{[1,\beta]}}\leq \beta 
\end{equation}

In order to achieve Supplementary Equation \eqref{eq:boundiVn}, we see how $m_t$ can be chosen: 
\begin{align}
I(V^n;W^n)&=\sum_{i=1}^{t-1}m_i I(X^i;Y^i)+m_{t}I(X^t;Y^t)\nonumber\\
       &\quad+\alpha I(X^{t+1};Y^{t+1})+I(X^{t+1}_{[1,\beta]};Y^{t+1}_{[1,\beta]})\\
       &\geq \sum_{i=1}^{t}m_in_i\left(\valA-\frac{\delta}{2^{i}}\right)+\alpha n_{t+1}\left(\valA-\frac{\delta}{2^{t+1}}\right)
\end{align}
In order to verify Supplementary Equation \eqref{eq:boundiVn} it suffices to choose $m_t$ larger than
\begin{equation}
\label{eq:mt}
	\left\lceil\frac{2^{t}}{n_t\delta}\left(\sum_{i=1}^{t-1}m_in_i\delta\left(\frac{1}{2^{i}}-\frac{1}{2^{t-1}}\right)  +n_{t+1}\left(\valA-\frac{\delta}{2^{t-1}}\right)\right)\right\rceil
\end{equation}
such that the following holds
\begin{equation}
	\sum_{i=1}^{t}m_in_i\left(\valA-\frac{\delta}{2^{i}}\right)\geq\left(\sum_{i=1}^tm_in_i+n_{t+1}\right)\left(\valA-\frac{\delta}{2^{t-1}}\right)
\end{equation}
However, for technical reasons in the concentration bounds that follow we choose:
\begin{equation}
\label{eq:mtmax}
	m_t=\max\left\{\textrm{Suplementary Equation }\eqref{eq:mt},(n_{t+1})^2\right\}
\end{equation}

In the following we prove that $\forall\eta>0$:
\begin{equation}
\label{eq:mokuteki}
\lim_{n\rightarrow\infty}\textrm{Pr}\left[\left|\frac{i_{V^n,W^n}}{nC_n}-1\right|\geq\eta\delta\right]=0
\end{equation}
Let us expand the probability expression in Supplementary Equation \eqref{eq:mokuteki}:
\begin{align}
	\textrm{Pr}&\left[\left|\frac{i_{V^n,W^n}}{nC_n}-1\right|\geq\eta\delta\right]=\nonumber\\
      &=\textrm{Pr}\left[\frac{i_{V^n,W^n}}{nC_n}-1\geq\eta\delta\right]
 +\textrm{Pr}\left[\frac{i_{V^n,W^n}}{nC_n}-1\leq-\eta\delta\right]\\
 &\leq\textrm{Pr}\left[\frac{i_{V^n,W^n}}{nC_n}-\mathbb E\left[\frac{i_{V^n,W^n}}{nC_n}\right]\geq\eta\delta\right]\nonumber\\
	   &\quad+\textrm{Pr}\left[\frac{i_{V^n,W^n}}{nC_n}-\mathbb E\left[\frac{i_{V^n,W^n}}{nC_n}\right]\leq-\eta\delta+\frac{\delta}{\valA 2^{t-1}}\right]\\
    &\leq\textrm{Pr}\left[\left|i_{V^n,W^n}-\mathbb E\left[i_{V^n,W^n}\right]\right|\geq nC_n\delta\left(\eta-\frac{1}{\valA 2^{t-1}}\right)\right]\label{eq:sumI}
\end{align}
Now we will exploit that $i_{V^n,W^n}$ can be expressed as a sum of $l=\sum_{i=1}^tm_i+\alpha+1$ independent random variables (see Supplementary Equation \eqref{eq:tnrelation}). For these sums we can bound the two-tailed probability via Hoeffding's inequality \cite{Hoeffding_63}. More concretely, let $\{X_i\}_{i=1}^l$ be a sequence of $l$ independent random variables, let $t\geq 0$ and let $a_i\leq X_i\leq b_i$ then:
\begin{equation}
\textrm{Pr}\left[\left|\sum_{i=1}^lX_i-\mathbb E\left[\sum_{i=1}^lX_i\right]\right|\geq t\right]\leq 2 \exp\left(\frac{-2t^2}{\sum_{i=1}^l|b_i-a_i|^2}\right)
\end{equation}
We can make clearly the identifications with Supplementary Equation \eqref{eq:sumI} and Supplementary Equation \eqref{eq:iid-information-spectrum}. However, before applying Hoeffding's inequality let us bound the denominator in the exponential term:
\begin{align}
	\sum_{i=1}^l|b_i-a_i|^2&=\sum_{i=1}^{t}m_i(n_i)^2+\alpha(n_{t+1})^2+\beta^2\\
	   &\leq\sum_{i=1}^{t}m_in_in_{t+1}+\alpha n_{t+1}n_{t+1}+\beta n_{t+1}\\
	     &\leq n^{3/2}
\end{align}
The last inequality follows because from Supplementary Equation \eqref{eq:mtmax} we have that $(n_{t+1})^2\leq m_{t}\leq n$. Now if we apply Hoeffding's inequality to Supplementary Equation \eqref{eq:sumI} we can bound it from above by:
\begin{equation}
2\exp\left(-\frac{2\left(nC_n\delta\left(\eta-\frac{1}{\valA 2^{t-1}}\right)\right)^2}{n^{3/2}}\right)
\label{eq:almosthere}
\end{equation}
and in consequence the limit when $n$ goes to infinity is zero for all $\eta>0$.
\end{proof}

\section{Proof of Lemma \ref{mainlemma1}}\label{proofs}

Lemma \ref{mainlemma1} is essentially proven by Gimbert and Oualhadj in \cite{Gimbert_09} with a very elegant construction (a succinct sketch can be found in \cite{Gimbert_10}). We include a full proof here for completeness, to cover the case of an arbitrary $\lambda$ (in \cite{Gimbert_09} they only consider the case $\lambda=1$) and to include in the construction an undecidability result of Hirvensalo \cite{Hirvensalo_07}. This allows us to give the concrete estimates of alphabet size $10$ and $62$ states that appear in Lemma \ref{mainlemma1}.

\subsection{The construction of Gimbert and Oualhadj}
\label{subsec:gimbert}

\begin{figure}
\centering
\includegraphics[width=7cm]{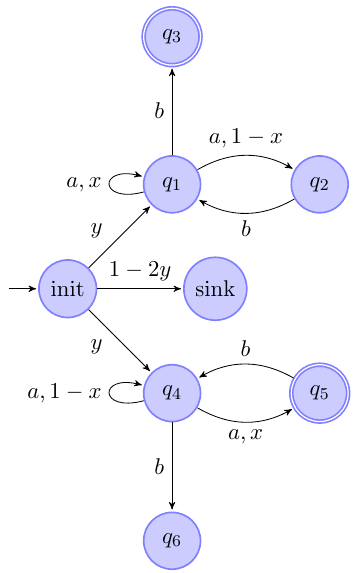}
\caption{The automaton $\D_{x,y}$ has value $2y$ if  $x>1/2$ and value $\le y$ if $x\le 1/2$.}\label{fig:ax}
\end{figure}
\begin{lemma}[Proposition 5 \cite{Gimbert_10}] 
\label{lem:gimbert5}
Let $\D_{x,y}$ be the automaton in \figurename~\ref{fig:ax} and $x\in[0,1]$, $y\in[0,1/2]$. $\D_{x,y}$ has value $2y$ if $x>1/2$ and value $\le y$ if $x\le 1/2$.
\end{lemma}
\begin{proof}
First, we need to make some observations regarding $\D_{x,y}$. If the input letter $b$ is fed two or more consecutive times the automaton is forced it into the states sink, $q_3$ and $q_6$ from which the automaton cannot exit. For any such a word, the acceptance value is $y$. Hence, we concentrate our attention to words of the form $a^{n_1}ba^{n_2}b\ldots ba^{n_t}b$. 
For any word $\word{w}$ of this form the acceptance value is:
\begin{align}
\val(\A,w)&=y\,p\left[q_1\xrightarrow{\,\word{w}\,} q_3\right]+y\,p\left[q_4\xrightarrow{\,\word{w}\,} q_5\right]\\
               &\leq y\,p\left[q_1\xrightarrow{\,\word{w}\,} q_3\right]+y\,\left(1-p\left[q_4\xrightarrow{\,\word{w}\,} q_6\right]\right)
\end{align}
Furthermore, the upper bound is reachable. To verify this, consider the word $\word{w}a^n$, $p\left[q_1\xrightarrow{\,\word{w}a^n\,} q_3\right]$ does not change and we can make $p\left[q_4\xrightarrow{\,\word{w}a^n\,} q_5\right]$ approach $1-p\left[q_4\xrightarrow{\,\word{w}a^n\,} q_6\right]$ by choosing $n$ large enough. Both $p\left[q_1\xrightarrow{\,\word{w}\,} q_3\right]$ and $p\left[q_4\xrightarrow{\,\word{w}\,} q_6\right]$ admit a very compact form:
\begin{align}
p\left[q_1\xrightarrow{\,\word{w}\,} q_3\right] &= 1-\prod_{i=1}^t (1-x^{n_i})\\
p\left[q_4\xrightarrow{\,\word{w}\,} q_6\right] &= 1-\prod_{i=1}^t (1-(1-x)^{n_i})
\end{align}

 Let us consider first $x\leq 1/2$. This implies that $x\leq1-x$ and in consequence 
\begin{equation}
1-\prod_{i=1}^t (1-x^{n_i})\leq 1-\prod_{i=1}^t (1-(1-x)^{n_i})\ .  
\end{equation}
Let $\epsilon>0$, for any word $\word{w}$ such that $p[q_1\rightarrow q_3]= 1-\epsilon$ we have $p[q_4\rightarrow q_6]\geq 1-\epsilon$  and $\val(\D_{x,y},\word{w}) \leq y$.

Let us assume now that $x>1/2$. We are going to prove that for any $\epsilon\in(0,x)$ there exists a word $\word{w}$ such that:
\begin{align}
p\left[q_4\xrightarrow{\,\word{w}\,}  q_6\right] &\leq \epsilon\label{eq:firstreq}\\
p\left[q_1\xrightarrow{\,\word{w}\,}  q_3\right] &\geq 1-\epsilon \label{eq:secondreq}
\end{align}
Consider the sequence of words $\{\word{w_k}\}_{k=2}^\infty$ where $\word{w_k}=a^{n_2}ba^{n_3}n\ldots ba^{n_k}$ and the lengths $n_2\ldots n_k$ are given by 
\begin{equation}
\label{eq:lengths}
n_k = \left\lceil \log_x\frac{1}{k}+C_\epsilon\right\rceil
\end{equation}
and
\begin{equation}
C_\epsilon=\frac{1}{b}\log_x\left(\frac{b-1}{b}\epsilon\right)\ .
\end{equation}
Let $b>1$ be a number such that $x^b=1-x$. The following sequence of inequalities holds:
\begin{align}
p\left[q_4\xrightarrow{\,\word{w_k}\,} q_6\right] &= (1-x)^{n_1}+(1-(1-x)^{n_2})(1-x)^{n_3}+\dots\nonumber\\
                                  &\quad +\prod_{i=2}^{k-1}(1-(1-x)^{n_i})(1-x)^{n_k}\\
                                  &\leq \sum_{i=2}^k (1-x)^{n_i}\\
                                  &= \sum_{i=2}^k x^{bn_i}\\
                                  &= \sum_{i=2}^k x^{b\lceil \log_x\frac{1}{i}+C_\epsilon\rceil}\\
                                  &\leq x^{bC_\epsilon}  \sum_{i=2}^k x^{b\log_x\frac{1}{i}}\\
                                  &=x^{bC_\epsilon}  \sum_{i=2}^k \frac{1}{i^b}\label{eq:rhs}
\end{align}
Note that the sum in the right hand side of Supplementary Equation \eqref{eq:rhs} when $k$ goes to infinity is very similar to the Riemann zeta function evaluated at a real argument strictly larger than one. For these arguments it is well known \cite{Jameson_03} that it can be bounded by
\begin{equation}
\zeta(b)=\sum_{n=1}^\infty \frac{1}{n^{b}}\leq \frac{b}{b-1}\ .
\end{equation} 
If we apply this bound to Supplementary Equation \eqref{eq:rhs} we obtain
\begin{align}                                  
\lim_{k\rightarrow\infty} p\left[q_4\xrightarrow{\,\word{w_k}\,} q_6\right] &\leq \lim_{k\rightarrow\infty} x^{bC_\epsilon}  \sum_{i=2}^k \frac{1}{i^b}\label{eq:riemzeta}\\
                                  &\leq x^{bC_\epsilon} \frac{b}{b-1}\\
                                  &=\epsilon\label{eq:eps}\ .
\end{align}
Furthermore, Supplementary Equation \eqref{eq:eps} remains an upper bound for finite $k$ since we are only dropping positive contributions. Hence, Supplementary Equation \eqref{eq:firstreq} is verified for all $k$. Let us now verify that there exists $k$ such that the  requirement Supplementary Equation \eqref{eq:secondreq} also holds. Consider the following sum
\begin{align}
\sum_{i=2}^k x^{n_i}&\geq \sum x^{\log_x\frac{1}{i}+C_\epsilon+1}\\
                     &= x^{C_\epsilon+1}\sum_{i=2}^k\frac{1}{i}
\end{align}
and this sum diverges for any non-zero $x$ and finite $C_\epsilon$. This implies that $\lim_{k\rightarrow\infty}\prod_{i=2}^k(1-x^{n_i})=0$ and that there exists a finite $k$ such that $\prod_{i=2}^k(1-x^{n_i})\leq\epsilon$. Then,  $p\left[q_1\xrightarrow{\,\word{w_k}\,} q_3\right] \geq 1-\epsilon$. 
\end{proof}

Now, we are going to modify $\D_{x,y}$. The main idea is that $x$ will be replaced by the probability that an automaton $\A$ accepts a word $\word{w}_\A$. This is achieved very easily, see \figurename~\ref{fig:dat}, once the state of $\D_{\A,y}$ reaches $\A$ it continues inside the automaton until it sees $c$ which is a symbol outside the input alphabet of $\A$. Then, it will transition to one of two different states depending on whether or not $\A$ is in an accepting state. We indicate the transitions from an accepting state by\includegraphics{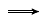} and the transitions from a non-accepting state by \includegraphics{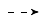}. Let $\word{w}_\A$ be an arbitrary input word into $\A$ then: 
\begin{align}
p\left[q_1\xrightarrow{a\word{w}_\A c}q_1\right]=\val(\A,\word{w}_\A) \\
p\left[q_4\xrightarrow{a\word{w}_\A c}q_5\right]=\val(\A,\word{w}_\A)
\end{align}

In the following we reduce the problem of finding the value of $\D_{\A,y}$ to the emptiness of the set $L_{\A>\lambda}$. This is the set of words with acceptance probability strictly higher than $\lambda$. That is: $L_{\A>\lambda}=\{\word{w} \in \mathcal W^*:\val(\A,\word{w})>\lambda\}$. 

\begin{figure}
\centering
\includegraphics[width=7.5cm]{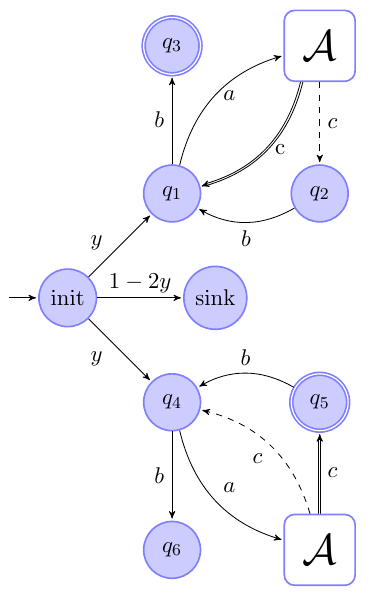}
\caption{The automaton $\D_{\A,y}$ has value $\le y$ if $L_{\A>1/2}$ is empty and value $\ge 2y$ if not.}\label{fig:dat}
\end{figure}

\begin{lemma}\label{lem:construction} Given a PFA $\mathcal{A}$ and $y\in [0,1/2]$, the automaton $\D_{\A,y}$ has value $\le y$ if $L_{\A>1/2}$ is empty and value $\ge 2y$ if not. 
\end{lemma}

\begin{proof}
 Assume first that $L_{\A>1/2}$ is not empty. Then there exists some $\word{w}_\A$ such that $\val(\A,\word{w}_\A)>1/2$. 
Hence, we can construct the sequence $\word{w_k}=(a\word{w}_\A c)^{n_2}\ldots (a\word{w}_\A c)^{n_k}$ with the lengths $n_2\ldots n_k$ given by Supplementary Equation \eqref{eq:lengths}. Following the proof of Lemma \ref{lem:gimbert5} we have that for $\epsilon>0$ there exists $k$ such that $\word{w_k}$ verifies conditions Supplementary Equation \eqref{eq:firstreq} and Supplementary Equation \eqref{eq:secondreq}.

Assume now that $L_{\A>1/2}$ is empty. We can restrict our attention to words of the form $(a\word{w}_{\A}^1 c)^{n_1}b\ldots b(a\word{w}_\A^k c)^{n_k}b$. Furthermore for any word $\word{w}$ we have that $\val(\A,\word{w})\leq 1-\val(\A,\word{w})$ and in consequence
\begin{equation}
1-\prod_{i=1}^k\left(1-\val(\A,\word{w}_{\A}^i)^{n_i}\right)\leq 1-\prod_{i=1}^k\left(1-\left(1-\val(\A,\word{w}_{\A}^i)^{n_i}\right)\right)
\label{eq:last}
\end{equation}
Let $\epsilon>0$ Supplementary Equation \eqref{eq:last} implies that for any word such that $p\left[q_1\xrightarrow{\word{w}}q_4\right]=1-\epsilon$ we have that $p\left[q_4\xrightarrow{\word{w}} q_6\right]\geq 1-\epsilon$  and $\val(\D_{\A,y},\word{w}) \leq y$.
\end{proof}

Let us close this section by defining the family $\mathcal{T}_{\lambda}$ as 

\begin{align}\label{def-family-T}
\mathcal{T}_{\lambda}=&\left\{\gamma(\D_{\A,\lambda/2}):\text{$\A$ has a binary alphabet and 27 states}\right\}
\end{align}
with $\gamma$ as in Definition \ref{def-gamma}.

\subsection{The undecidability result of Hirvensalo}
\label{subsec:hirvensalo}

To use the above construction in order to prove Lemma \ref{mainlemma1}, up to the issue of restricting to \stable\ and resettable channels (that we will take care of below), it only remains to show that deciding whether $L_{\A>1/2}$ is empty or not is indeed undecidable. 
This problem, known as the emptiness problem, was proved undecidable in \cite{Ginsburg_66,Paz_71,Condon_89}. Recently, new proofs with explicit bounds in the number of states and the cardinality of the alphabet have been derived in \cite{Blondel_03,Hirvensalo_07,Gimbert_10} together with an undecidability proof of several related sets. Here, we will rely on  

\begin{theorem}[\cite{Hirvensalo_07}]
\label{th:emproblem}
Let $k$ be an integer equal or greater than 7 and $(n,m)$ be a duple of integers that is equal or pointwise larger than $(2,5k-10)$. The emptiness of $L_{\A>\delta}$, for  
$\delta=1/(5k-10)$ and PFAs with alphabet size $n$ and $m$ states, is undecidable. \end{theorem}

\begin{figure}
\centering
\begin{tabular}{cc}
\subfloat[$\B_p$]{\includegraphics[width=4cm]{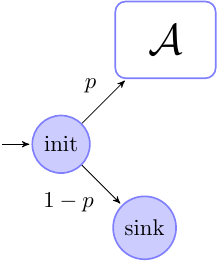}} 
   & \subfloat[$\C_p$]{\includegraphics[width=4cm]{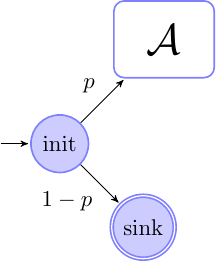}}
\end{tabular}
\caption{The automata $\B_p$ (left) and $\C_p$ (right) can be used to amplify the undecidability of the emptiness problem to arbitrary $\delta\in(0,1)$.}\label{fig:pfamulti}
\end{figure}

Taking Theorem \ref{th:emproblem} as a starting point, we can amplify the result and obtain undecidability for any rational $\delta\in(0,1)$ (in particular for $\delta=1/2$). 

\begin{corollary}\label{cor-Hirvensalo}
Fix any rational number $\delta$. The emptiness of $L_{\A>\delta}$ for PFAs with alphabet size $2$ and $27$ states is undecidable.
\end{corollary}

\begin{proof}
Given an arbitrary PFA $\A=(\Q,\W,\X,v,\F)$ and $p\in(0,1)$ we are going to construct two PFAs $\B_p$ and $\C_p$ such that: $L_{\A>\delta}$ is empty $\Leftrightarrow$ $L_{\B_p>p\delta}$ is empty $\Leftrightarrow$ $L_{\C_p>p\delta+1-p}$ is empty. 

Let us first construct $\B_p=(\T,\W,\Y,u,\F)$. The set of states is $\T=\{\Q\cup\textrm{init}\cup\textrm{sink}\}$. The input alphabet is equal to the original one. 
For any input symbol $x\in\W$ we define the stochastic matrices of $\B_p$ as follows:
\begin{equation}
Y_x=
\left(
\begin{tabular}{c c c|c c }
& & & &  0\\
& $X_x$ & & $p\, X_x\, v $&$\vdots$\\
& & & & 0\\
\hline
0&$\ldots$ &0& 0 & 0 \\
0&$\ldots$ &0& 1-p & 1
\end{tabular}\right)
\end{equation}
Note that we have added two rows and columns to track the two new states. Let us parse the action of the automaton as defined by the stochastic matrices. If it is in any of the original states, its behavior remains unchanged. If the automaton is in the sink state no matter what input symbol it reads the PFA remains in the sink state. Finally, if the automaton is in the init state upon reading the input symbol $x$ with probability $1-p$ it will transition to the sink state and with probability $p$ it will transition to whatever the original automaton would have transitioned from the initial distribution. In other words, the new distribution on the states will be given by $\left(p\, X_x\, v, 0, 1-p\right)$. The initial distribution of $\B_p$ has weight one on the init state, that is: $u=\left(0, \ldots, 0, 1, 0\right)$.

The construction of $\C_p$ is identical except that we add the sink state to the set of accepting states. We have depicted both constructions in \figurename~\ref{fig:pfamulti}.

For any input word $\word{w}\in\mathcal W^*$ we have that $\val(\A,\word{w})=p\,\val(\B_p,\word{w})=p\,\val(\C_p,\word{w})+1-p$. Hence, $L_{\A>\delta}$ is empty $\Leftrightarrow$ $L_{\B_p>p\delta}$ is empty $\Leftrightarrow$ $L_{\C_p>p\delta+1-p}$ is empty. 
\end{proof}

\subsection{Resettable and \stable\ channels}
\label{subsec:resandstab}
By the definition of the family $\mathcal{T}_\lambda$ given in (\ref{def-family-T}), Lemma \ref{mainlemma1} is just a consequence of Lemma \ref{lem:construction}, Corollary \ref{cor-Hirvensalo} and the following lemma, whose proof finishes the paper.

\begin{lemma}
\label{lem:reslemma}
$\val_{\tilde{\A}}= \val_{\gamma(\tilde{\A})}$ for all PFA $\tilde{A}$ of the form $D_{\A,y}$.
\end{lemma}

\begin{proof}
Given a  PFA $\A$, we define the set $\textrm{values}(\A)=\{\val(\A,\word{w})|\word{w}\in \W^*\}$. This is the set of achievable values or, alternatively, it can be regarded as the range of the function $\val(\A,\word{w})$ once the PFA $\A$ is fixed. It is then enough to show that $\textrm{values}(\tilde{\A})=\textrm{values}(\gamma(\tilde{\A}))$ for any PFA $\tilde{A}$ of the form $D_{\A,y}$. 

$\supseteq$

This direction is trivial since any input word $\word{w}$ of $\tilde{\A}$ is also an input word of $\gamma(\tilde{\A})$ and $\val(\gamma(\tilde{\A}),\word{w})=\val(\tilde{\A},\word{w})$. 

$\subseteq$

Let us divide the input words into two sets: $W_1$ the words that either end with the symbol $\rt$ or consist of a string of $\id$ and $W_2$ which is the complementary set, that is, words that have at least one symbol different than $\id$ and do not end with the $\rt$ symbol. The acceptance probability of any $\word{w}\in W_1$ is simply the acceptance probability of a distribution with unit probability on the initial symbol. Since for $\D_{\tilde{\A},y}$
the acceptance and initial symbols are disjoint, the value of $\word{w}$ is zero. That means that no word from $W_1$ can be in the set $\{\word{w}: \val(\gamma(\tilde{\A}),\word{w})\geq\lambda\}$ for any value of $\lambda\in(0,1]$.

First, consider any word $\word{w}\in W_2$ that contains at least one identity symbol, it can be written as $\word{w_1}\id \word{w_2}$ where $\word{w_1}$ and $\word{w_2}$ are two sequences of input symbols and at least one of both is non empty. We have that $\val(\gamma(\tilde{\A}),\word{w})=\val(\gamma(\tilde{\A}),\word{w_1}\word{w_2})$ and by applying this argument to all the identity symbols in the word we find a new word $\word{w'}$ with no identity symbols such that $\val(\gamma(\tilde{\A}),\word{w})=\val(\gamma(\tilde{\A}),\word{w'})$. Hence we can restrict our attention to words with no identity symbol. 

Second, we consider any word  $\word{w}\in W_2$ that contains at least one reset symbol, it can be written as $\word{w_1}\rt \word{w_2}$ where at least $\word{w_2}$ is non empty. We have that $\val(\gamma(\tilde{\A}),\word{w})=\val(\gamma(\tilde{\A}),\word{w_2})$, again we can apply this argument to all the reset symbols in the word and find a word $\word{w'}$ with no reset or identity symbols such that $\val(\gamma(\tilde{\A}),\word{w})=\val(\gamma(\tilde{\A}),\word{w'})=\val(\tilde{\A},\word{w'})$. 
\end{proof}


%

\end{document}